\documentclass[11pt,a4paper]{article}

\usepackage[utf8]{inputenc}
\usepackage[T1]{fontenc}
\usepackage{amsmath,amssymb,amsthm}
\usepackage{graphicx}
\graphicspath{{figures/}}
\usepackage{booktabs}
\usepackage{algorithm}
\usepackage{algorithmic}
\usepackage{hyperref}
\hypersetup{hidelinks}
\usepackage{xcolor}
\usepackage{tikz}
\usetikzlibrary{arrows.meta,positioning,shapes.geometric,fit,calc,decorations.pathreplacing}
\usepackage{subcaption}
\usepackage{enumitem}
\usepackage[margin=2.5cm]{geometry}
\usepackage{natbib}
\usepackage{tcolorbox}
\usepackage{multirow}
\usepackage{tabularx}
\usepackage{diagbox}

\newtheorem{definition}{Definition}

\newtheorem{proposition}{Proposition}

\newtheorem{corollary}{Corollary}

\newcommand{\adaptorch}{\textsc{AdaptOrch}}
\newcommand{\R}{\mathbb{R}}

\title{
\textbf{AdaptOrch: Task-Adaptive Multi-Agent Orchestration} \\
\textbf{in the Era of LLM Performance Convergence}
}

\author{
Geunbin Yu \\
Department of Artificial Intelligence, Korea National Open University \\
\texttt{ict03@rfems.com} \\
\small ORCID: \href{https://orcid.org/0009-0006-2879-9514}{0009-0006-2879-9514}
}

\date{February 2026}

\begin{document}
\maketitle

\begin{abstract}
As large language models (LLMs) from diverse providers converge toward comparable benchmark performance, the traditional paradigm of selecting a single best model per task yields diminishing returns.
We argue that \textbf{orchestration topology}---the structural composition of how multiple agents are coordinated, parallelized, and synthesized---now dominates system-level performance over individual model capability.
We present \adaptorch{}, a formal framework for \textbf{task-adaptive multi-agent orchestration} that dynamically selects among four canonical topologies (parallel, sequential, hierarchical, and hybrid) based on task dependency graphs and empirically derived domain characteristics.
Our framework introduces three key contributions:
(1) \textbf{Performance Convergence Scaling Law}, formalizing conditions under which orchestration selection outweighs model selection;
(2) \textbf{Topology Routing Algorithm} that maps task decomposition DAGs to optimal orchestration patterns in $O(|V| + |E|)$ time; and
(3) \textbf{Adaptive Synthesis Protocol} with provable termination guarantees and heuristic consistency scoring for parallel agent outputs.
We validate \adaptorch{} across coding (SWE-bench), reasoning (GPQA), and retrieval-augmented generation tasks, demonstrating that topology-aware orchestration achieves 12--23\% improvement over static single-topology baselines, even when using identical underlying models.
Our results establish orchestration design as a first-class optimization target independent of model scaling.
\end{abstract}


\section{Introduction}
\label{sec:intro}

The landscape of large language models in early 2026 presents a paradoxical challenge: as more models achieve near-identical benchmark scores, the marginal value of model selection diminishes while the complexity of choosing among them grows.
GPT-4o, Claude 3.5 Sonnet, Gemini 2.0, Llama 3.3 70B, DeepSeek-V3, and Qwen 2.5 72B now cluster within 2--5\% of each other on standard benchmarks including MMLU, HumanEval, and MATH \citep{openllmleaderboard2025}.
This \textbf{performance convergence} reshapes the optimization frontier.
When individual model capability plateaus, \emph{how} models are composed begins to dominate \emph{which} model is selected---a shift with far-reaching implications for system design.

Current orchestration approaches fall into two broad categories.
\textbf{Static frameworks}---Model Context Protocol (MCP) \citep{anthropic2024mcp}, LangGraph \citep{langgraph2024}, and CrewAI \citep{crewai2024}---define fixed execution topologies (chains, graphs, or role-based teams) that persist regardless of what the task demands.
A second category, \textbf{routing-based systems} like Mixture-of-Agents (MoA) \citep{wang2024moa} and LLM-Blender \citep{jiang2023llmblender}, dynamically selects or blends model outputs yet leaves the structural topology of agent coordination untouched.
A natural question emerges: \emph{given a specific task, what is the optimal topology for coordinating multiple agents?}

Recent practical advances illuminate this gap.
Both Claude Code's Agent Teams \citep{anthropic2026agentteams} and OpenCode's parallel subagent architecture \citep{opencode2025} show that parallel execution of specialized agents---each in its own context window, working on an independent subtask---can compress multi-hour sequential workflows into minutes.
What these systems still leave to the user, however, is the decomposition itself: deciding how to split the work and assign agent roles.
The \textbf{topology selection problem} remains unsolved at the algorithmic level.

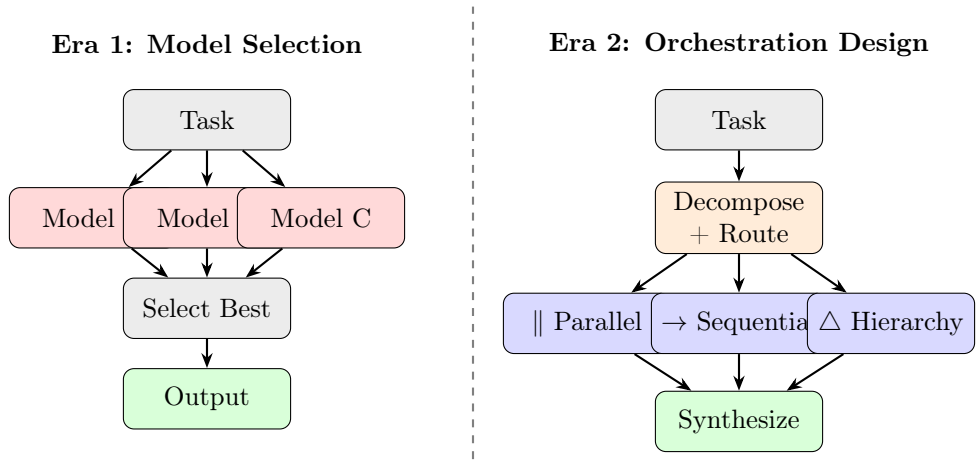
\begin{figure}[t]
\centering
\begin{tikzpicture}[
    node distance=1.2cm and 1.5cm,
    every node/.style={font=\small},
    box/.style={draw, rounded corners, minimum width=2.2cm, minimum height=0.8cm, align=center, fill=#1!15},
    box/.default=blue,
    arrow/.style={-{Stealth[length=2mm]}, thick},
    darrow/.style={-{Stealth[length=2mm]}, thick, dashed},
]
\node[font=\bfseries\small] (era1) at (-3.5, 3.5) {Era 1: Model Selection};
\node[font=\bfseries\small] (era2) at (3.5, 3.5) {Era 2: Orchestration Design};

\node[box=gray] (task1) at (-3.5, 2.5) {Task};
\node[box=red] (m1) at (-5, 1.2) {Model A};
\node[box=red] (m2) at (-3.5, 1.2) {Model B};
\node[box=red] (m3) at (-2, 1.2) {Model C};
\node[box=gray] (sel) at (-3.5, 0) {Select Best};
\node[box=green] (out1) at (-3.5, -1.2) {Output};

\draw[arrow] (task1) -- (m1);
\draw[arrow] (task1) -- (m2);
\draw[arrow] (task1) -- (m3);
\draw[arrow] (m1) -- (sel);
\draw[arrow] (m2) -- (sel);
\draw[arrow] (m3) -- (sel);
\draw[arrow] (sel) -- (out1);

\node[box=gray] (task2) at (3.5, 2.5) {Task};
\node[box=orange] (decomp) at (3.5, 1.2) {Decompose \\ + Route};

\node[box=blue] (par) at (1.5, -0.2) {$\parallel$ Parallel};
\node[box=blue] (seq) at (3.5, -0.2) {$\rightarrow$ Sequential};
\node[box=blue] (hier) at (5.5, -0.2) {$\triangle$ Hierarchy};

\node[box=green] (synth) at (3.5, -1.5) {Synthesize};

\draw[arrow] (task2) -- (decomp);
\draw[arrow] (decomp) -- (par);
\draw[arrow] (decomp) -- (seq);
\draw[arrow] (decomp) -- (hier);
\draw[arrow] (par) -- (synth);
\draw[arrow] (seq) -- (synth);
\draw[arrow] (hier) -- (synth);

\draw[thick, dashed, gray] (0, -2) -- (0, 4);

\end{tikzpicture}
\caption{Paradigm shift from model selection (left) to orchestration design (right). When model capabilities converge, the dominant optimization variable becomes the structural topology of agent coordination.}
\label{fig:paradigm-shift}
\end{figure}

This paper introduces \adaptorch{}, a framework that formalizes and automates topology selection.
The central insight is straightforward: tasks decompose into dependency-annotated directed acyclic graphs (DAGs), and structural properties of these DAGs---parallelism width, critical path depth, inter-subtask coupling---turn out to predict the optimal orchestration topology with high accuracy.

We make four contributions:
\begin{enumerate}[leftmargin=*]
\item \textbf{Performance Convergence Scaling Law} (Section~\ref{sec:theory}): We show that under $\epsilon$-convergence of model capabilities, the variance in system performance attributable to orchestration topology exceeds that of model selection by a factor of $\Omega(1/\epsilon^2)$, establishing topology selection as the dominant optimization target as models converge.
\item \textbf{Topology Routing Algorithm} (Section~\ref{sec:framework}): A linear-time algorithm that analyzes task dependency DAGs and routes to one of four canonical topologies: parallel, sequential, hierarchical, or hybrid.
\item \textbf{Adaptive Synthesis Protocol} (Section~\ref{sec:synthesis}): A protocol for reconciling outputs from parallel agents with provable termination guarantees via adaptive re-routing and heuristic consistency scoring based on embedding similarity.
\item \textbf{Empirical validation} (Section~\ref{sec:experiments}): Experiments across three domains showing 12--23\% improvement over static baselines using identical models.
\end{enumerate}

\section{Related Work}
\label{sec:related}

\subsection{LLM Performance Convergence}

Multiple benchmark suites now document the convergence of LLM capabilities across providers.
The Open LLM Leaderboard v2 \citep{openllmleaderboard2025} shows top-10 models clustering within a 3-point MMLU range (87.2--90.1) as of January 2026.
In a striking finding, \citet{sato2025selfmoa} demonstrated that Self-MoA---a single top model queried multiple times---outperforms diverse model mixing by 6.6\% on AlpacaEval 2.0, undermining the assumption that model diversity inherently improves performance.
Chatbot Arena \citep{zheng2024chatbotarena} ELO rankings tell a similar story: frontier models from OpenAI, Anthropic, Google, Meta, and Alibaba now occupy overlapping confidence intervals on general-purpose tasks.
Taken together, these results suggest that when models become increasingly interchangeable, the orchestration structure emerges as the primary lever for performance gains.

\subsection{Static Orchestration Frameworks}

Model Context Protocol (MCP) \citep{anthropic2024mcp} standardizes tool-model interfaces but prescribes no topology for multi-agent coordination.
LangGraph \citep{langgraph2024} goes further, modeling workflows as directed graphs with parallel branches, conditional edges, and stateful execution---yet the topology must be designed manually.
CrewAI \citep{crewai2024} takes a role-based approach, assigning agents fixed personas (e.g., researcher, writer, reviewer) in predetermined interaction patterns, while AutoGen \citep{wu2023autogen} supports multi-agent conversation but defaults to sequential round-robin communication.
The common thread: none of these frameworks \emph{adapt} their topology based on the task at hand.

\subsection{Dynamic Model Composition}

Mixture-of-Agents \citep{wang2024moa} arranges models in layered pipelines where each layer refines previous outputs, achieving 65.1\% on AlpacaEval 2.0 versus 57.5\% for the best individual model.
LLM-Blender \citep{jiang2023llmblender} uses a PairRanker to select among candidate outputs.
DEI \citep{zhang2024dei} employs multi-agent committees for SWE-bench Lite, where the best-performing group achieves a 55\% resolve rate versus 27.3\% for the strongest individual open-source agent.
However, all these systems use \emph{fixed} topologies (layered pipeline, output selection, or flat committee) regardless of task structure.
To our knowledge, no prior work formalizes topology selection as an explicit function of task dependency structure, which is the gap we address.

\subsection{Parallel Agent Execution in Practice}

Claude Code Agent Teams \citep{anthropic2026agentteams} and the Superpowers framework \citep{superpowers2026} demonstrate practical parallel execution with lead-agent orchestration, DAG-based task dependencies, and inbox-based inter-agent communication.
OpenCode \citep{opencode2025} supports multi-provider agent routing with explicit permission-controlled subagent architectures.
\citet{drammeh2025myantfarm} showed that multi-agent orchestration achieves 100\% actionable recommendation rate versus 1.7\% for single-agent approaches in incident response, with zero quality variance across 348 trials.
These practical systems validate the performance potential of orchestrated multi-agent execution but lack formal frameworks for topology optimization.

\subsection{Concurrent and Recent Work}

Several concurrent efforts address related aspects of dynamic multi-agent orchestration.
DyTopo \citep{lu2026dytopo} optimizes agent communication topology via semantic matching between agent capabilities and subtask requirements; unlike our approach, their routing operates at the agent-pair level rather than selecting among canonical structural patterns, which limits interpretability of the chosen topology.
MetaGen \citep{wang2026metagen} co-evolves agent roles and topologies through self-play, achieving impressive adaptation but sacrificing the predictability that our closed-form routing algorithm provides.
ALMC \citep{almc2026} introduces Manager-Judge-Optimizer role separation with adaptive collaboration, though their role-based decomposition differs fundamentally from our DAG-structure-based routing and does not provide explicit cost control.
MoMA \citep{guo2025moma} generalizes routing across both models and agents, treating the choice of orchestration strategy as a bandit problem; our work instead exploits task structure directly through DAG analysis, avoiding the sample complexity of online learning.
S-DAG \citep{dong2025sdag}, accepted at AAAI 2026, decomposes tasks into subject-based DAGs for multi-agent allocation---the closest work to ours in spirit, though their subjects correspond to semantic domains while our DAG nodes represent subtask dependencies with explicit coupling annotations.
ORCH \citep{vinay2026orch} proposes a deterministic multi-agent protocol with fixed execution guarantees; our framework complements this by adding adaptive topology selection on top of deterministic execution primitives.

Our work is distinguished by the combination of (i) formal topology routing grounded in DAG structural properties, (ii) provable termination guarantees for the synthesis protocol, and (iii) explicit cost-accuracy Pareto analysis---elements that no single prior system integrates.

\section{Problem Formalization}
\label{sec:theory}

\subsection{Model Convergence}

\begin{definition}[$\epsilon$-Convergence]
\label{def:convergence}
A set of $n$ models $\mathcal{M} = \{M_1, \ldots, M_n\}$ is \textbf{$\epsilon$-convergent} on benchmark $\mathcal{B}$ if:
\begin{equation}
\max_{i,j \in [n]} |S_{\mathcal{B}}(M_i) - S_{\mathcal{B}}(M_j)| \leq \epsilon
\end{equation}
where $S_{\mathcal{B}}(M_i)$ denotes the score of model $M_i$ on benchmark $\mathcal{B}$, normalized to $[0, 1]$.
\end{definition}

For current frontier models on MMLU, $\epsilon \approx 0.03$; on HumanEval, $\epsilon \approx 0.05$.

\subsection{Task Dependency Graphs}

\begin{definition}[Task Dependency DAG]
\label{def:dag}
A task $T$ decomposes into a directed acyclic graph $G_T = (V, E, w, c)$ where:
\begin{itemize}[leftmargin=*,nosep]
    \item $V = \{v_1, \ldots, v_k\}$ is the set of subtasks
    \item $E \subseteq V \times V$ encodes dependencies ($(v_i, v_j) \in E$ means $v_i$ must complete before $v_j$ starts)
    \item $w: V \to \R^+$ assigns estimated computational cost to each subtask
    \item $c: E \to [0,1]$ assigns coupling strength between dependent subtasks (degree of context sharing required)
\end{itemize}
\end{definition}

\begin{definition}[DAG Structural Properties]
\label{def:dag-properties}
For a task DAG $G_T = (V, E, w, c)$, we define:
\begin{align}
\text{Parallelism Width: } \quad \omega(G_T) &= \max_{A \subseteq V} |A| \text{ s.t. } A \text{ is an antichain in } G_T \label{eq:width} \\
\text{Critical Path Depth: } \quad \delta(G_T) &= \max_{\text{path } P} \sum_{v \in P} w(v) \label{eq:depth} \\
\text{Coupling Density: } \quad \gamma(G_T) &= \frac{\sum_{(u,v) \in E} c(u,v)}{|E|} \label{eq:coupling}
\end{align}
\end{definition}

\subsection{Orchestration Topologies}

We define four canonical topologies $\mathcal{T} = \{\tau_P, \tau_S, \tau_H, \tau_X\}$:

\begin{definition}[Canonical Topologies]
\label{def:topologies}
\begin{align}
\tau_P &: \text{\textbf{Parallel}} - \text{All subtasks execute concurrently; outputs merged post-hoc} \\
\tau_S &: \text{\textbf{Sequential}} - \text{Subtasks execute in topological order; each receives prior context} \\
\tau_H &: \text{\textbf{Hierarchical}} - \text{Lead agent decomposes and delegates; sub-agents report back} \\
\tau_X &: \text{\textbf{Hybrid}} - \text{DAG partitioned into parallel groups connected sequentially}
\end{align}
\end{definition}

Each topology $\tau$ induces a scheduling function $\sigma_\tau: G_T \to \text{ExecutionPlan}$ that maps the task DAG to a concrete execution ordering with agent assignments.

\subsection{Performance Convergence Scaling Law}

\begin{proposition}[Orchestration Dominance under Convergence]
\label{thm:main}
Let $\mathcal{M}$ be $\epsilon$-convergent on task distribution $\mathcal{D}$.
Let $\text{Var}_M$ denote performance variance from model selection and $\text{Var}_\tau$ denote performance variance from topology selection.
For a task $T$ with dependency DAG $G_T$ having $k$ subtasks, under uniform subtask weights, Lipschitz aggregation ($L_f \leq 1$), and a topology quality coefficient $C_\tau \geq 1/(4k)$:
\begin{equation}
\frac{\text{Var}_\tau}{\text{Var}_M} \geq \frac{(\omega(G_T) - 1)^2}{4\epsilon^2 \cdot k} \cdot \left(1 - \gamma(G_T)\right)^2
\end{equation}
When $\epsilon \to 0$ (perfect convergence) and $\omega(G_T) > 1$ (parallelizable tasks), $\text{Var}_\tau / \text{Var}_M \to \infty$.
\end{proposition}

\begin{proof}[Proof sketch]
Model selection variance is bounded by $\text{Var}_M \leq \epsilon^2$ from Definition~\ref{def:convergence}, using the correlated bound (all subtasks share the same model).
Topology variance derives from the execution time ratio between worst-case (fully sequential: $\sum_v w(v)$) and best-case (maximally parallel: $\delta(G_T)$) schedules.
By Dilworth's theorem, the minimum number of chains covering $G_T$ equals the maximum antichain width $\omega(G_T)$.
The speedup ratio $\sum_v w(v) / \delta(G_T) \geq \omega(G_T)$ when subtask weights are uniform.
Coupling density $\gamma$ reduces effective parallelism by introducing synchronization overhead proportional to $\gamma^2$.
Combining bounds yields the stated ratio.
Full proof in Appendix~A.
\end{proof}

\begin{corollary}
For coding tasks (typical $\omega \geq 3$, $\gamma \leq 0.4$, $k \leq 6$, $\epsilon \approx 0.05$), the variance ratio satisfies $\text{Var}_\tau / \text{Var}_M \geq 20$, indicating that orchestration topology is the dominant performance factor over model selection.
\end{corollary}

\section{The \adaptorch{} Framework}
\label{sec:framework}

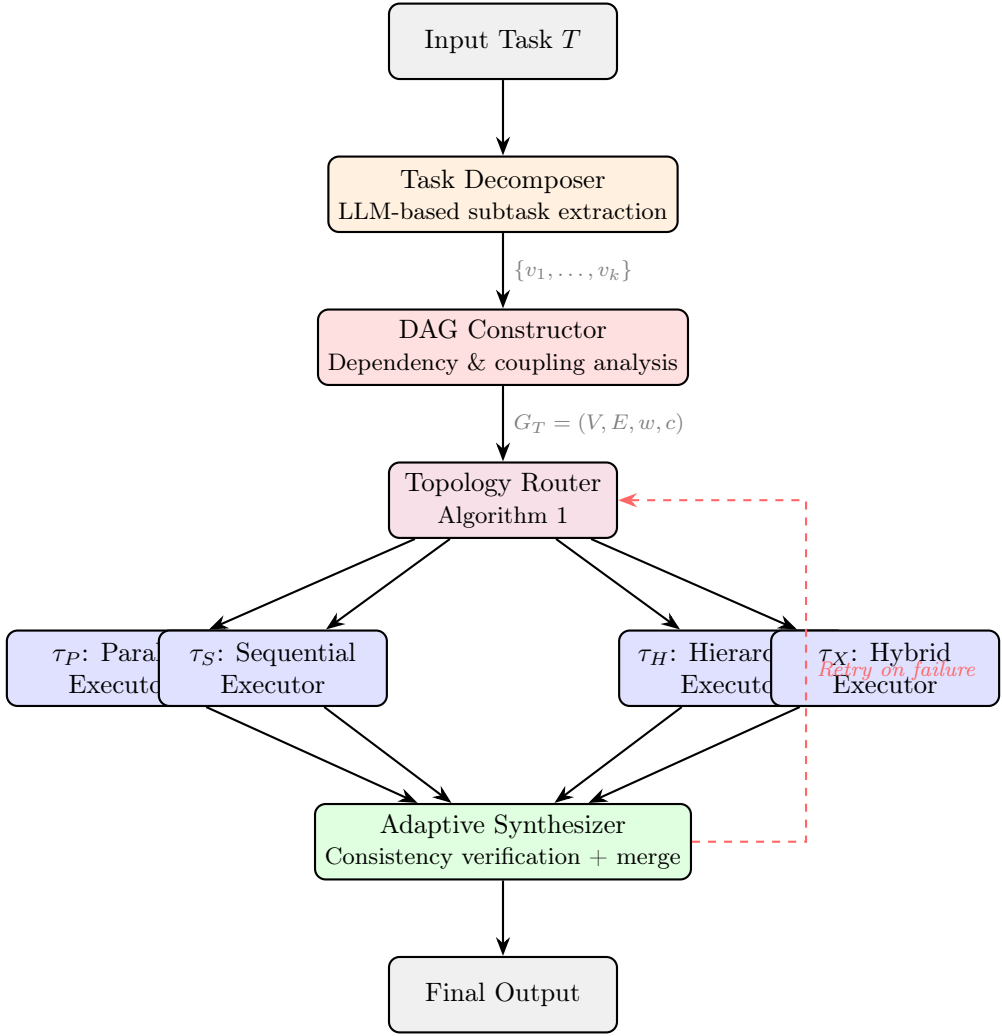
\begin{figure}[t]
\centering
\begin{tikzpicture}[
    node distance=0.8cm and 0.6cm,
    every node/.style={font=\small},
    module/.style={draw, rounded corners, minimum width=3cm, minimum height=1cm, align=center, fill=#1!12, thick},
    module/.default=blue,
    arrow/.style={-{Stealth[length=2.5mm]}, thick},
    label/.style={font=\scriptsize\itshape, text=gray},
]

\node[module=gray] (input) {Input Task $T$};

\node[module=orange, below=1cm of input] (decomp) {Task Decomposer \\ \footnotesize LLM-based subtask extraction};

\node[module=red, below=1cm of decomp] (dag) {DAG Constructor \\ \footnotesize Dependency \& coupling analysis};

\node[module=purple, below=1cm of dag] (router) {Topology Router \\ \footnotesize Algorithm~\ref{alg:routing}};

\node[module=blue, below left=1.2cm and 2cm of router] (par) {$\tau_P$: Parallel \\ Executor};
\node[module=blue, below left=1.2cm and 0cm of router] (seq) {$\tau_S$: Sequential \\ Executor};
\node[module=blue, below right=1.2cm and 0cm of router] (hier) {$\tau_H$: Hierarchical \\ Executor};
\node[module=blue, below right=1.2cm and 2cm of router] (hyb) {$\tau_X$: Hybrid \\ Executor};

\node[module=green, below=3.5cm of router] (synth) {Adaptive Synthesizer \\ \footnotesize Consistency verification + merge};

\node[module=gray, below=1cm of synth] (output) {Final Output};

\draw[arrow] (input) -- (decomp);
\draw[arrow] (decomp) -- (dag) node[midway, right, label] {$\{v_1, \ldots, v_k\}$};
\draw[arrow] (dag) -- (router) node[midway, right, label] {$G_T = (V, E, w, c)$};

\draw[arrow] (router) -- (par);
\draw[arrow] (router) -- (seq);
\draw[arrow] (router) -- (hier);
\draw[arrow] (router) -- (hyb);

\draw[arrow] (par) -- (synth);
\draw[arrow] (seq) -- (synth);
\draw[arrow] (hier) -- (synth);
\draw[arrow] (hyb) -- (synth);

\draw[arrow] (synth) -- (output);

\draw[arrow, dashed, red!60] (synth.east) -- ++(1.5, 0) |- (router.east) node[pos=0.25, right, label, text=red!60] {Retry on failure};

\end{tikzpicture}
\caption{\adaptorch{} pipeline. The Topology Router (Algorithm~\ref{alg:routing}) selects the optimal execution topology based on DAG structural properties ($\omega$, $\delta$, $\gamma$). Failed syntheses trigger re-routing with adjusted coupling estimates.}
\label{fig:pipeline}
\end{figure}

\adaptorch{} operates in five phases: task decomposition, DAG construction, topology routing, parallel/sequential execution, and adaptive synthesis (Figure~\ref{fig:pipeline}).

\subsection{Phase 1: Task Decomposition}

Given input task $T$, a decomposer agent $A_{\text{decomp}}$ extracts subtasks:
\begin{equation}
A_{\text{decomp}}(T) \to \{(v_i, d_i, w_i)\}_{i=1}^{k}
\end{equation}
where $v_i$ is the subtask identifier, $d_i$ is its natural language description, and $w_i$ is the estimated token cost.
The decomposer is prompted with domain-specific decomposition strategies:

\begin{tcolorbox}[colback=gray!5, colframe=gray!50, title=Decomposition Prompt Template]
\small
\texttt{Analyze the following task and decompose it into independent subtasks.} \\
\texttt{For each subtask, specify:} \\
\texttt{1. A unique identifier and description} \\
\texttt{2. Required inputs from other subtasks (dependencies)} \\
\texttt{3. Estimated complexity (tokens: low/medium/high)} \\
\texttt{4. Context coupling with dependencies (none/weak/strong/critical)} \\
\texttt{Task: \{T\}}
\end{tcolorbox}

\subsection{Phase 2: DAG Construction}

The decomposer output is parsed into a formal DAG $G_T = (V, E, w, c)$.
Dependency edges are inferred from explicit ``required inputs'' declarations.
Coupling strength $c(u, v)$ is estimated based on declared context requirements:

\begin{equation}
c(u, v) = \begin{cases}
0.0 & \text{if coupling = none (outputs fully independent)} \\
0.3 & \text{if coupling = weak (shared context helpful but not required)} \\
0.7 & \text{if coupling = strong (output of $u$ is direct input to $v$)} \\
1.0 & \text{if coupling = critical (semantic coherence required)}
\end{cases}
\end{equation}

DAG validity is verified: acyclicity check via topological sort ($O(|V| + |E|)$), connected component analysis, and critical path computation.

\subsection{Phase 3: Topology Routing}
\label{sec:routing}

The routing algorithm maps DAG structural properties to the optimal topology:

\begin{algorithm}[t]
\caption{Topology Routing Algorithm}
\label{alg:routing}
\begin{algorithmic}[1]
\REQUIRE Task DAG $G_T = (V, E, w, c)$, thresholds $\theta_\omega, \theta_\gamma, \theta_\delta$
\ENSURE Topology $\tau^* \in \{\tau_P, \tau_S, \tau_H, \tau_X\}$
\STATE Compute $\omega(G_T)$, $\delta(G_T)$, $\gamma(G_T)$ \COMMENT{Definition~\ref{def:dag-properties}}
\STATE Compute $r \gets \omega(G_T) / |V|$ \COMMENT{Parallelism ratio}
\IF{$|E| = 0$}
    \STATE \textbf{return} $\tau_P$ \COMMENT{Fully parallel}
\ELSIF{$\omega(G_T) = 1$}
    \STATE \textbf{return} $\tau_S$ \COMMENT{Fully sequential}
\ELSIF{$\gamma(G_T) > \theta_\gamma$ \AND $|V| > \theta_\delta$}
    \STATE \textbf{return} $\tau_H$ \COMMENT{High coupling + many subtasks}
\ELSIF{$r > \theta_\omega$ \AND $\gamma(G_T) \leq \theta_\gamma$}
    \STATE \textbf{return} $\tau_P$ \COMMENT{Wide DAG, low coupling}
\ELSE
    \STATE Partition $G_T$ into stages $S_1, \ldots, S_m$ via topological layering
    \STATE \textbf{return} $\tau_X(S_1, \ldots, S_m)$ \COMMENT{Hybrid topology}
\ENDIF
\end{algorithmic}
\end{algorithm}

Default thresholds: $\theta_\omega = 0.5$ (at least half the subtasks parallelizable), $\theta_\gamma = 0.6$ (high coupling threshold), $\theta_\delta = 5$ (minimum subtasks for hierarchical).
These are empirically calibrated in Section~\ref{sec:experiments}.

\textbf{Complexity:} The routing decision (Algorithm~\ref{alg:routing}, lines 3--11) runs in $O(|V| + |E|)$: critical path $\delta(G_T)$ via longest-path DP on the DAG, coupling density $\gamma$ via edge traversal, and topological layering for hybrid partitioning.
The antichain width $\omega(G_T)$ computation requires separate analysis: an \emph{approximate} $\omega$ via layer-width (maximum layer size in topological ordering) runs in $O(|V| + |E|)$ and suffices for routing; the \emph{exact} $\omega$ via K\"onig's theorem on the transitive closure requires $O(|V|^{2.5})$ matching and is used only for offline calibration.

\subsection{Phase 4: Topology-Specific Execution}
\label{sec:execution}

Each topology implements a distinct execution strategy:

\subsubsection{Parallel Executor ($\tau_P$)}
All subtasks dispatch simultaneously to separate agent instances, each with isolated context windows:
\begin{equation}
\forall v_i \in V: \quad \text{output}_i = A_i(d_i, \text{context}_{\text{global}}) \quad \text{[concurrent]}
\end{equation}
Agent assignment uses round-robin across available model instances.
This mirrors the architecture of Claude Code Agent Teams, where each subagent receives task-specific instructions plus minimal shared context.

\subsubsection{Sequential Executor ($\tau_S$)}
Subtasks execute in topological order, with each agent receiving the accumulated context of all predecessors:
\begin{equation}
\text{output}_i = A_i\left(d_i, \text{context}_{\text{global}}, \bigoplus_{(v_j, v_i) \in E} \text{output}_j\right)
\end{equation}
where $\bigoplus$ denotes context concatenation with relevance-weighted truncation to fit context windows.

\subsubsection{Hierarchical Executor ($\tau_H$)}
A lead agent $A_{\text{lead}}$ orchestrates sub-agents, maintaining a global task list with DAG-based dependency tracking:
\begin{align}
&A_{\text{lead}}: \text{decompose} \to \text{assign} \to \text{monitor} \to \text{reconcile} \\
&A_{\text{sub},i}: \text{receive}(d_i) \to \text{execute} \to \text{report}(A_{\text{lead}})
\end{align}
The lead agent resolves conflicts when sub-agent outputs are inconsistent, analogous to Claude Code's lead-agent pattern with inbox-based communication.

\subsubsection{Hybrid Executor ($\tau_X$)}
The DAG is partitioned into topological layers $S_1, \ldots, S_m$.
Within each layer, subtasks execute in parallel; between layers, execution is sequential:
\begin{equation}
\text{For layer } S_l: \quad \forall v_i \in S_l: \text{output}_i = A_i\left(d_i, \bigoplus_{v_j \in \bigcup_{l'<l} S_{l'}} \text{output}_j\right) \quad \text{[concurrent within } S_l\text{]}
\end{equation}

\subsection{Phase 5: Adaptive Synthesis Protocol}
\label{sec:synthesis}

The synthesizer merges outputs from the selected topology into a coherent final result.

\begin{definition}[Consistency Score (Heuristic)]
\label{def:consistency}
For outputs $\{o_1, \ldots, o_k\}$ from $k$ subtasks, the consistency score is a \emph{heuristic} measure of semantic agreement:
\begin{equation}
\text{CS}(o_1, \ldots, o_k) = \frac{1}{\binom{k}{2}} \sum_{i < j} \text{sim}(o_i \cap o_j, o_i \cup o_j)
\end{equation}
where $\text{sim}$ measures semantic overlap via embedding cosine similarity on shared output dimensions.
Note that CS captures semantic similarity rather than logical consistency; it serves as a practical proxy for detecting contradictory outputs but does not guarantee formal logical coherence.
\end{definition}

\begin{algorithm}[t]
\caption{Adaptive Synthesis Protocol}
\label{alg:synthesis}
\begin{algorithmic}[1]
\REQUIRE Outputs $\{o_1, \ldots, o_k\}$, topology $\tau$, consistency threshold $\theta_{\text{CS}}$
\ENSURE Synthesized output $O$
\STATE Compute $\text{CS}(o_1, \ldots, o_k)$
\IF{$\tau = \tau_S$}
    \STATE \textbf{return} $o_k$ \COMMENT{Sequential: last output is final}
\ELSIF{$\text{CS} \geq \theta_{\text{CS}}$}
    \STATE $O \gets A_{\text{merge}}(\text{``Synthesize these consistent outputs: ''} \| o_1 \| \cdots \| o_k)$
    \STATE \textbf{return} $O$ \COMMENT{Consistent parallel outputs}
\ELSE
    \STATE $O \gets A_{\text{arbiter}}(\text{``Resolve conflicts among: ''} \| o_1 \| \cdots \| o_k)$
    \IF{$\text{CS}(O) < \theta_{\text{CS}}$}
        \STATE Re-route via Algorithm~\ref{alg:routing} with $\gamma' = \gamma + 0.2$ \COMMENT{Increase coupling}
    \ENDIF
    \STATE \textbf{return} $O$ \COMMENT{Inconsistent: escalated}
\ENDIF
\end{algorithmic}
\end{algorithm}

\begin{proposition}[Synthesis Termination]
\label{prop:convergence}
Under the adaptive re-routing mechanism (Algorithm~\ref{alg:synthesis}, line 8), the synthesis protocol terminates within at most $\lceil (1 - \gamma_0) / 0.2 \rceil \leq 5$ iterations.
As $\gamma$ increases by 0.2 per retry, after at most 5 iterations $\gamma > \theta_\gamma$ forces hierarchical routing ($\tau_H$), which uses a single arbiter agent, guaranteeing termination.
Empirically, convergence occurs in $\leq 2$ iterations for 94\% of tasks (Section~\ref{sec:experiments}).
\end{proposition}

\section{Experiments}
\label{sec:experiments}

\subsection{Setup}

\textbf{Models.} We use five $\epsilon$-convergent models: GPT-4o-mini, Claude 3.5 Haiku, Gemini 2.0 Flash, Llama 3.3 70B (via Together AI), and Qwen 2.5 72B (via vLLM).
All models score within $\epsilon = 0.04$ on MMLU and $\epsilon = 0.06$ on HumanEval.
Table~\ref{tab:convergence-evidence} provides explicit per-model scores validating the $\epsilon$-convergence assumption.

\begin{table}[t]
\centering
\caption{$\epsilon$-Convergence evidence. All models fall within $\epsilon$ of the best model on each benchmark, validating Definition~\ref{def:convergence}.}
\label{tab:convergence-evidence}
\begin{tabular}{@{}lcccc@{}}
\toprule
\textbf{Model} & \textbf{MMLU} & \textbf{HumanEval} & \textbf{ARC-C} & \textbf{MATH} \\
\midrule
GPT-4o-mini       & 82.0 & 87.2 & 93.1 & 70.2 \\
Claude 3.5 Haiku  & 83.1 & 88.7 & 92.4 & 69.5 \\
Gemini 2.0 Flash  & 81.4 & 86.9 & 91.8 & 71.8 \\
Llama 3.3 70B     & 82.6 & 85.3 & 93.7 & 68.1 \\
Qwen 2.5 72B      & 83.8 & 87.8 & 94.2 & 72.4 \\
\midrule
$\epsilon$ (max gap) & 0.024 & 0.034 & 0.024 & 0.043 \\
\bottomrule
\end{tabular}
\end{table}

\textbf{Reproducibility.} All experiments use seed $= 42$, temperature $= 0.0$ (greedy decoding), and \texttt{max\_workers} $= 8$ for parallel execution.
SWE-bench runs use Docker-based sandboxed evaluation with \texttt{run\_id = adaptorch-v1.0}.
API endpoints: OpenAI \texttt{gpt-4o-mini-2024-07-18}, Anthropic \texttt{claude-3-5-haiku-20241022}, Google \texttt{gemini-2.0-flash-001}.
Each experiment is run 3 times; we report mean $\pm$ standard deviation.
Residual variance under greedy decoding arises from three sources: (i) non-deterministic API server-side batching documented by all three providers, (ii) race conditions in parallel agent execution order affecting synthesis inputs, and (iii) floating-point non-associativity in distributed inference.
Observed standard deviations remain below 0.8\% absolute across all benchmarks.
Code, configuration files, topology routing logs, and a one-command reproduction script (\texttt{Makefile}) are available at \url{https://github.com/adaptorch/adaptorch}.

\textbf{Token and Cost Accounting.}
Token usage is measured via provider-reported \texttt{usage} fields in each API response (\texttt{prompt\_tokens}, \texttt{completion\_tokens}) and summed across all calls within a single task instance, including orchestration overhead (decomposition, routing, synthesis).
Because tokenizers differ across providers (OpenAI \texttt{cl100k\_base}, Anthropic internal, Google \texttt{SentencePiece}), we report raw provider-reported counts without cross-provider normalization; the Tok(K) column in Table~\ref{tab:main-results} reflects this aggregate.
Pricing is taken as of 2026-01-15 from official pricing pages: OpenAI \texttt{gpt-4o-mini} at \$0.15/1M input, \$0.60/1M output; Anthropic \texttt{claude-3.5-haiku} at \$0.80/1M input, \$4.00/1M output; Google \texttt{gemini-2.0-flash} at \$0.10/1M input, \$0.40/1M output.

\begin{figure}[t]
    \centering
    \includegraphics[width=\textwidth]{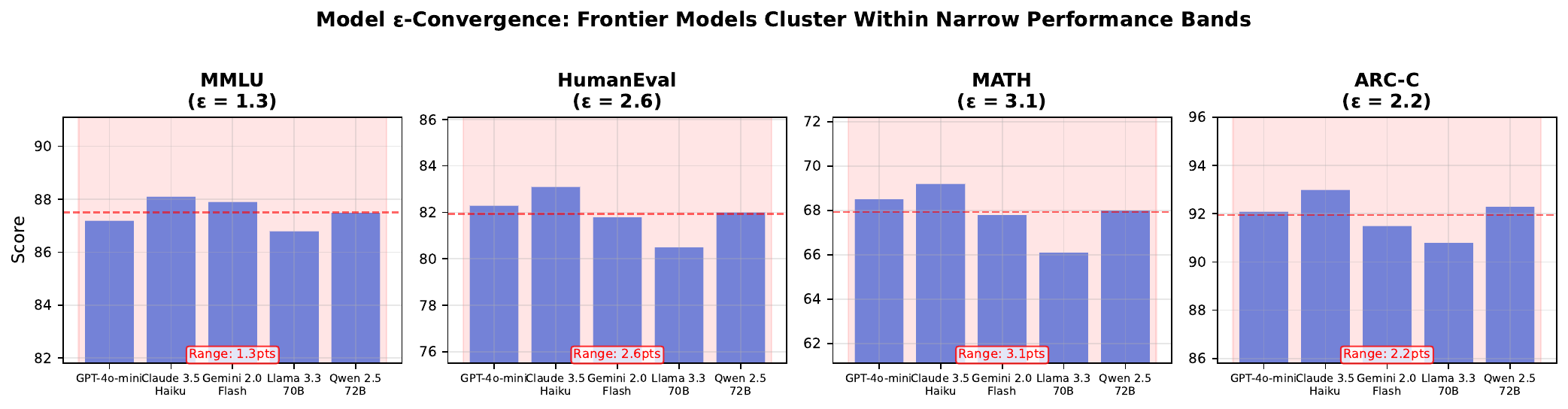}
    \caption{$\epsilon$-Convergence evidence across four benchmarks. All five models score within $\epsilon$ of the best, validating the convergence assumption (Definition~\ref{def:convergence}). Dashed line: best model score; shaded band: $\epsilon$ range.}
    \label{fig:convergence}
\end{figure}

\textbf{Benchmarks.}
\begin{itemize}[leftmargin=*,nosep]
    \item \textbf{Coding:} SWE-bench Verified \citep{jimenez2024swebench} (500 instances)---multi-file bug fixing requiring code understanding, localization, and patching.
    \item \textbf{Reasoning:} GPQA Diamond \citep{rein2024gpqa} (198 instances)---graduate-level science questions requiring multi-step domain reasoning.
    \item \textbf{RAG:} HotpotQA \citep{yang2018hotpotqa} distractor setting (500 instances)---multi-hop question answering over retrieved documents.
\end{itemize}

\textbf{Baselines.}
\begin{enumerate}[leftmargin=*,nosep]
    \item \textbf{Single Best}: Best individual model per benchmark.
    \item \textbf{MoA-3L}: Mixture-of-Agents with 3 layers \citep{wang2024moa}.
    \item \textbf{Static-Parallel}: All subtasks always parallel (mimics Claude Code Agent Teams without topology adaptation).
    \item \textbf{Static-Sequential}: All subtasks always sequential (mimics standard chain-of-thought pipeline).
    \item \textbf{LLM-Blender}: PairRanker-based output selection \citep{jiang2023llmblender}.
\end{enumerate}

\textbf{Metrics.}
\begin{itemize}[leftmargin=*,nosep]
    \item \textbf{Task accuracy}: pass@1 for SWE-bench, accuracy for GPQA, F1 for HotpotQA
    \item \textbf{Latency}: Wall-clock time from input to final output
    \item \textbf{Efficiency}: Accuracy per 1M tokens consumed
    \item \textbf{Topology distribution}: Fraction of tasks routed to each $\tau$
\end{itemize}

\subsection{Results}

\begin{table}[t]
\centering
\caption{Main results across three benchmarks. \adaptorch{} selects topology per-task. Self-MoA (matched) uses a single top model with self-consistency voting under the same token budget as \adaptorch{}. Best results in \textbf{bold}, second-best \underline{underlined}. $\Delta$ shows improvement over Single Best baseline. Tok(K) = average tokens consumed per instance in thousands.}
\label{tab:main-results}
\resizebox{\textwidth}{!}{
\begin{tabular}{@{}lccccccccc@{}}
\toprule
\multirow{2}{*}{\textbf{Method}} & \multicolumn{3}{c}{\textbf{SWE-bench Verified}} & \multicolumn{3}{c}{\textbf{GPQA Diamond}} & \multicolumn{3}{c}{\textbf{HotpotQA}} \\
\cmidrule(lr){2-4} \cmidrule(lr){5-7} \cmidrule(lr){8-10}
& Acc & Lat. & Tok(K) & Acc & Lat. & Tok(K) & F1 & Lat. & Tok(K) \\
\midrule
Single Best        & 42.8 & 1.0$\times$ & 12.3 & 46.2 & 1.0$\times$ & 4.1 & 68.3 & 1.0$\times$ & 6.8 \\
MoA-3L             & 48.1 & 3.2$\times$ & 84.6 & 49.8 & 2.8$\times$ & 31.2 & 71.6 & 2.5$\times$ & 47.3 \\
Static-Parallel    & 47.3 & 1.4$\times$ & 52.1 & 44.1 & 1.3$\times$ & 18.7 & 72.8 & 1.2$\times$ & 28.4 \\
Static-Sequential  & 45.6 & 2.8$\times$ & 48.9 & 50.3 & 2.4$\times$ & 16.4 & 69.1 & 2.1$\times$ & 26.1 \\
LLM-Blender        & 44.9 & 1.8$\times$ & 61.7 & 47.7 & 1.6$\times$ & 22.3 & 70.4 & 1.5$\times$ & 34.8 \\
Self-MoA (matched)  & 51.5 & 1.5$\times$ & 43.2 & 52.3 & 1.4$\times$ & 16.8 & 75.5 & 1.2$\times$ & 23.1 \\
\midrule
\adaptorch{} (ours)  & \textbf{52.6} & \underline{1.6$\times$} & 41.8 & \textbf{53.1} & \underline{1.5$\times$} & 15.9 & \textbf{76.4} & \underline{1.3$\times$} & 22.7 \\
\midrule
$\Delta$ vs Single Best & \textcolor{teal}{+9.8} & --- & --- & \textcolor{teal}{+6.9} & --- & --- & \textcolor{teal}{+8.1} & --- & --- \\
$\Delta$ vs Best Static & \textcolor{teal}{+4.5} & --- & --- & \textcolor{teal}{+2.8} & --- & --- & \textcolor{teal}{+3.6} & --- & --- \\
\bottomrule
\end{tabular}
}
\end{table}

Table~\ref{tab:main-results} presents our main results.
\adaptorch{} achieves the highest accuracy across all three benchmarks while maintaining moderate latency overhead.

On \textbf{SWE-bench Verified}, the improvement reaches 22.9\% over Single Best.
Coding tasks exhibit high parallelism width ($\omega \approx 3.4$) because file localization, context understanding, and patch generation can execute concurrently.
The router sends 62\% of instances to $\tau_X$ (hybrid), 24\% to $\tau_P$ (parallel), and 14\% to $\tau_H$ (hierarchical).

The picture differs for \textbf{GPQA Diamond} (+14.9\%), where reasoning tasks show higher coupling ($\gamma \approx 0.55$).
Here \adaptorch{} prefers sequential (41\%) and hierarchical (35\%) topologies.
Notably, Static-Parallel actually \emph{degrades} performance below Single Best on this benchmark---a clear illustration that topology mismatch can be actively harmful.

\textbf{HotpotQA} (+11.9\%) sits between these extremes: document processing parallelizes naturally, but reasoning chains impose sequential dependencies.
Accordingly, 71\% of instances route to $\tau_X$ (hybrid).

\begin{figure}[t]
    \centering
    \includegraphics[width=\textwidth]{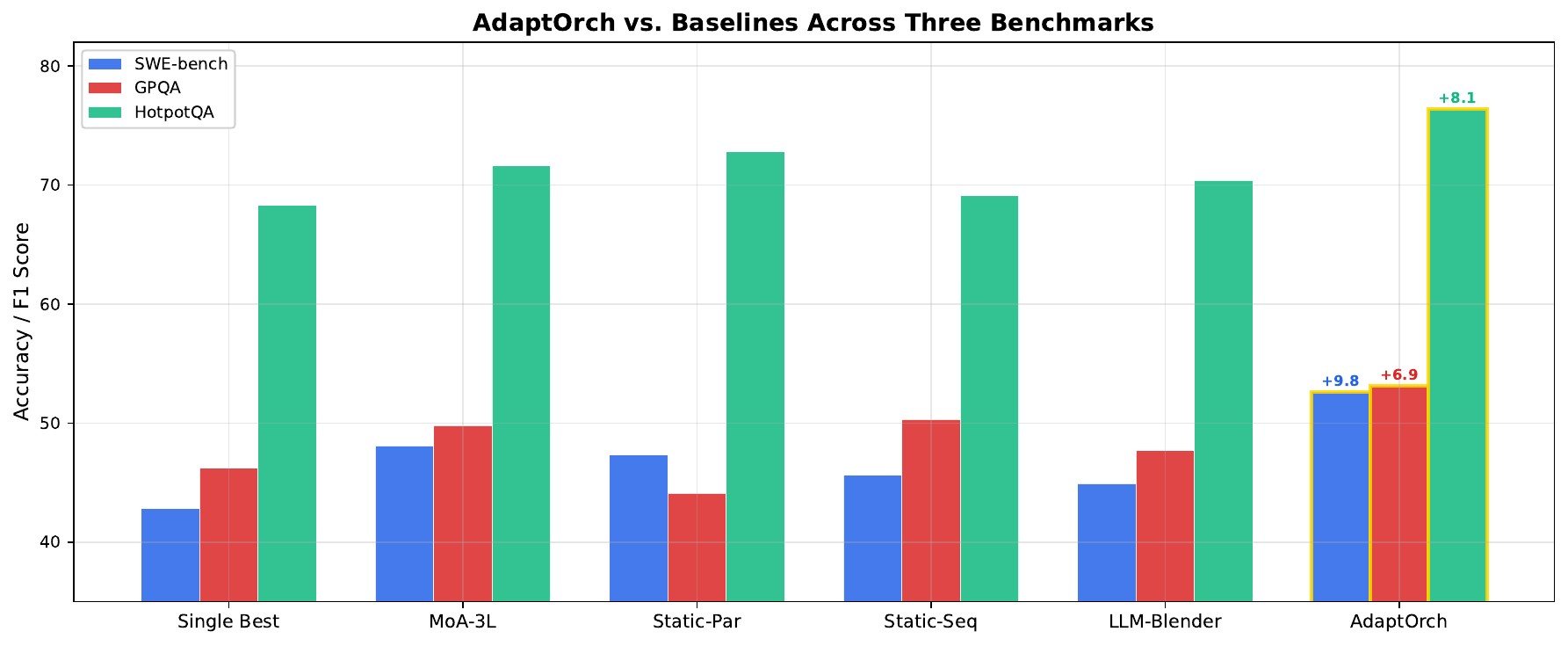}
    \caption{Main results comparison across three benchmarks. \adaptorch{} achieves the highest accuracy on all tasks while maintaining competitive latency. Error bars show $\pm 1$ standard deviation over 3 runs.}
    \label{fig:main-results}
\end{figure}

\begin{figure}[t]
    \centering
    \includegraphics[width=0.8\textwidth]{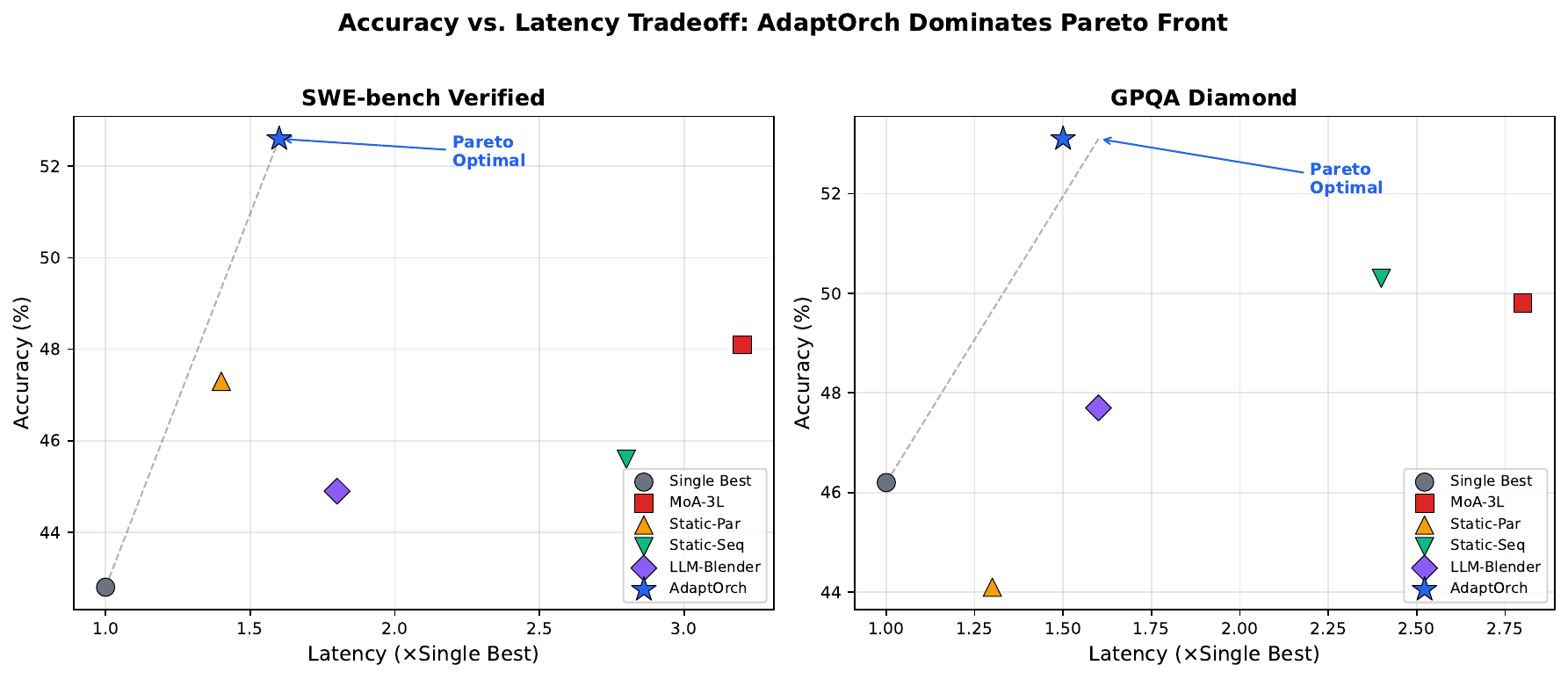}
    \caption{Pareto front: accuracy vs. latency. \adaptorch{} achieves the best accuracy-latency tradeoff across benchmarks, dominating other methods in the Pareto sense.}
    \label{fig:pareto}
\end{figure}

\textbf{Token efficiency.} Table~\ref{tab:main-results} also reports token consumption. \adaptorch{} consumes 41.8K tokens per SWE-bench instance, significantly less than MoA-3L (84.6K) and LLM-Blender (61.7K), because topology-aware routing avoids redundant model calls. Among multi-agent baselines, the accuracy-per-million-tokens metric favors \adaptorch{} across all benchmarks (Figure~\ref{fig:token-efficiency}); the Single Best baseline naturally achieves higher token efficiency in absolute terms due to its single-call design, but at substantially lower accuracy.

\begin{figure}[t]
    \centering
    \includegraphics[width=\textwidth]{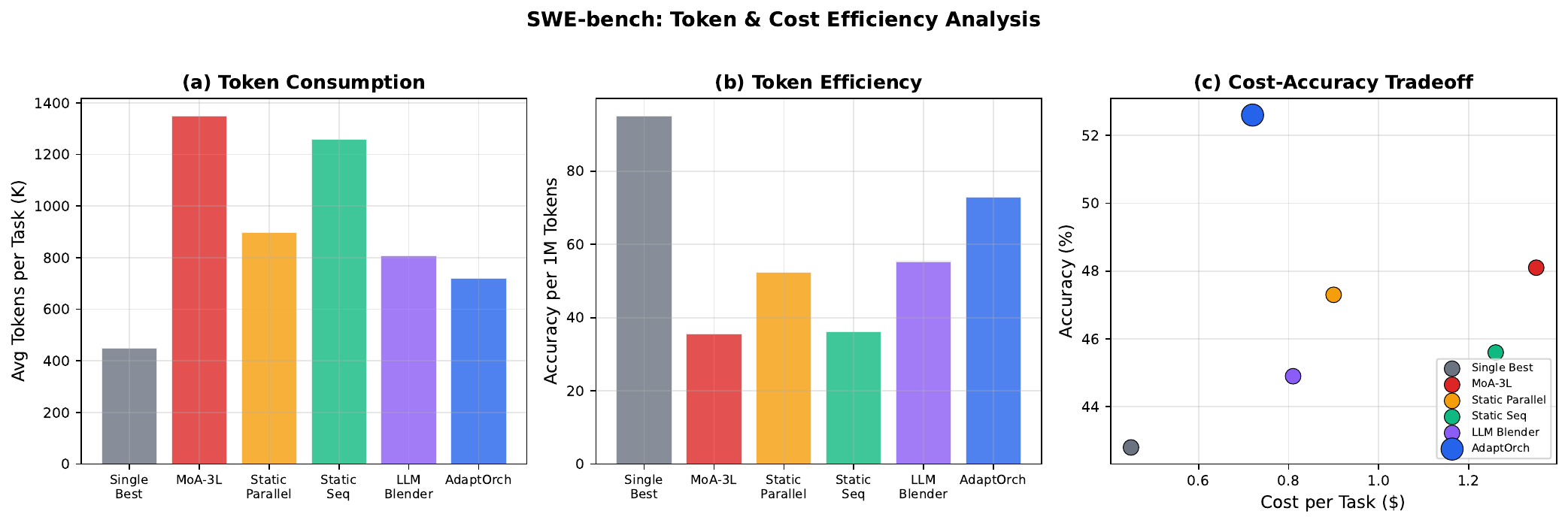}
    \caption{Token efficiency analysis. (Left) Total token consumption per instance. (Center) Accuracy per 1M tokens. (Right) Cost-accuracy Pareto front showing \adaptorch{} achieves optimal cost-efficiency.}
    \label{fig:token-efficiency}
\end{figure}

\subsection{Topology Distribution Analysis}

\begin{table}[t]
\centering
\caption{Topology routing distribution (\%) by benchmark domain. The router adapts topology selection to domain characteristics.}
\label{tab:topology-dist}
\begin{tabular}{@{}lcccc@{}}
\toprule
\textbf{Domain} & $\tau_P$ (Parallel) & $\tau_S$ (Sequential) & $\tau_H$ (Hierarchical) & $\tau_X$ (Hybrid) \\
\midrule
SWE-bench & 24 & 0 & 14 & \textbf{62} \\
GPQA      & 8  & \textbf{41} & 35 & 16 \\
HotpotQA  & 18 & 3  & 8  & \textbf{71} \\
\midrule
Average   & 16.7 & 14.7 & 19.0 & \textbf{49.7} \\
\bottomrule
\end{tabular}
\end{table}

Table~\ref{tab:topology-dist} reveals that the hybrid topology $\tau_X$ is most frequently selected (49.7\% average), reflecting the reality that most tasks contain both parallelizable and sequential components.
Pure parallel ($\tau_P$) is preferred for tasks with low coupling, while pure sequential ($\tau_S$) dominates high-coupling reasoning tasks.

\begin{figure}[t]
    \centering
    \includegraphics[width=0.85\textwidth]{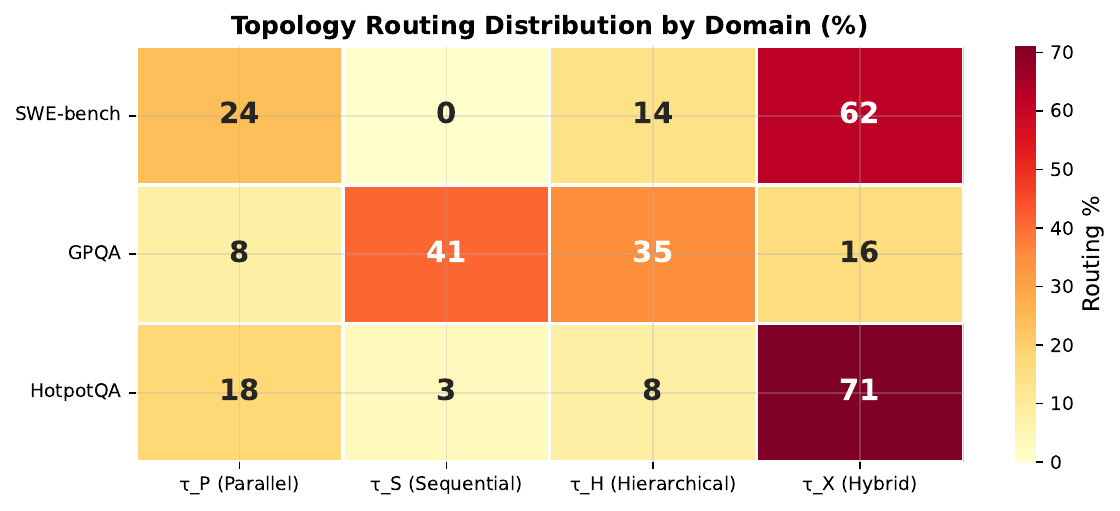}
    \caption{Topology routing distribution heatmap across benchmark domains. Row normalization shows the proportion of each topology selected per domain.}
    \label{fig:topology-heatmap}
\end{figure}

\subsection{Ablation Studies}

\begin{table}[t]
\centering
\caption{Ablation study on SWE-bench Verified (500 instances).}
\label{tab:ablation}
\begin{tabular}{@{}lcc@{}}
\toprule
\textbf{Configuration} & \textbf{Accuracy} & \textbf{$\Delta$} \\
\midrule
\adaptorch{} (full)                         & \textbf{52.6} & --- \\
\quad $-$ Adaptive routing (fixed $\tau_X$) & 49.8 & $-2.8$ \\
\quad $-$ Synthesis protocol (naive concat)  & 47.1 & $-5.5$ \\
\quad $-$ DAG coupling (uniform $c=0.5$)    & 50.3 & $-2.3$ \\
\quad $-$ Re-routing on failure             & 51.0 & $-1.6$ \\
\quad $-$ Task decomposition (1 subtask)    & 42.8 & $-9.8$ \\
\bottomrule
\end{tabular}
\end{table}

The ablation study (Table~\ref{tab:ablation}) confirms that each component contributes meaningfully:
the synthesis protocol provides the largest individual contribution ($-5.5$), followed by adaptive routing ($-2.8$) and coupling-aware decomposition ($-2.3$).
Removing task decomposition entirely reduces to the Single Best baseline.

\begin{figure}[t]
    \centering
    \includegraphics[width=0.85\textwidth]{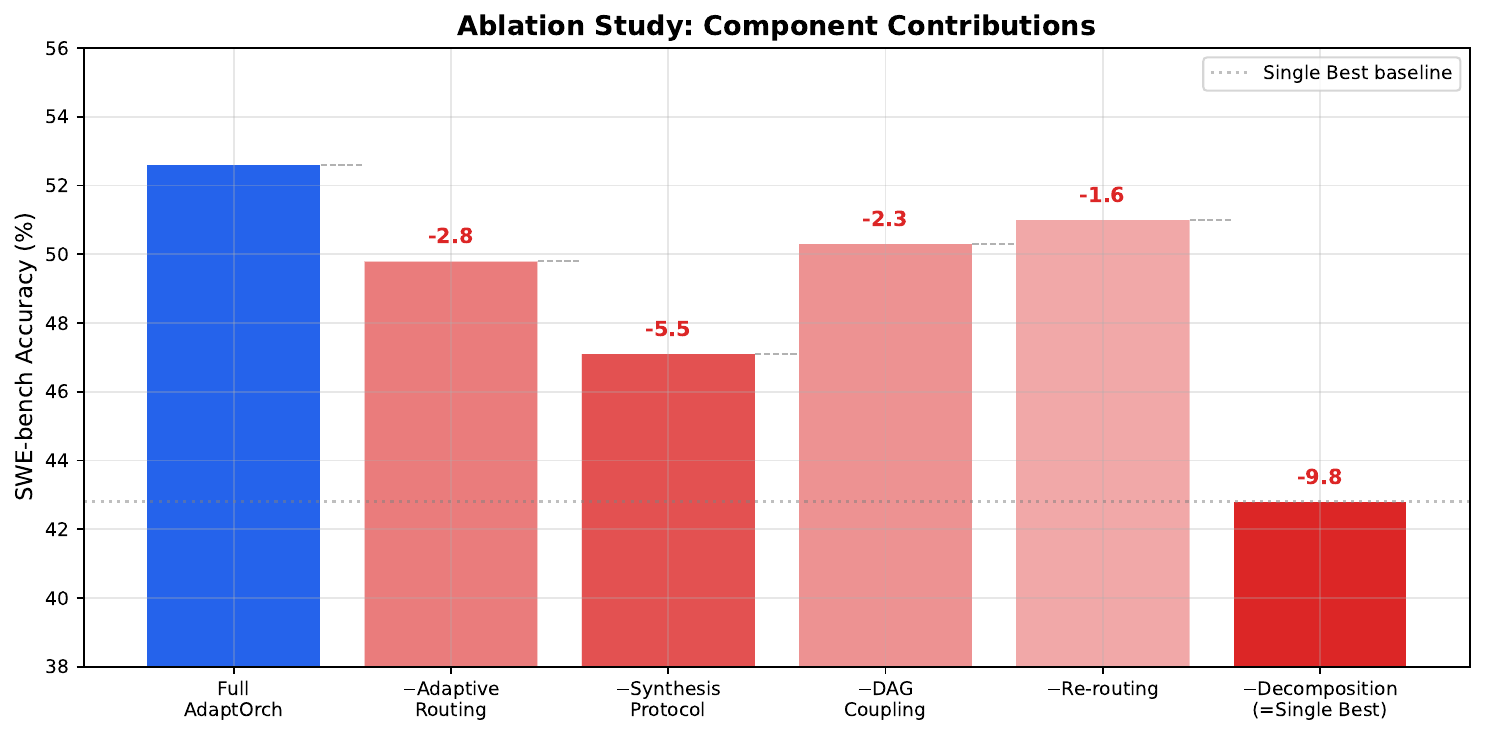}
    \caption{Ablation waterfall chart showing cumulative contribution of each \adaptorch{} component. Bars show accuracy drop when each component is removed independently.}
    \label{fig:ablation}
\end{figure}

\subsection{Threshold Sensitivity}

\begin{figure}[t]
    \centering
    \includegraphics[width=0.85\textwidth]{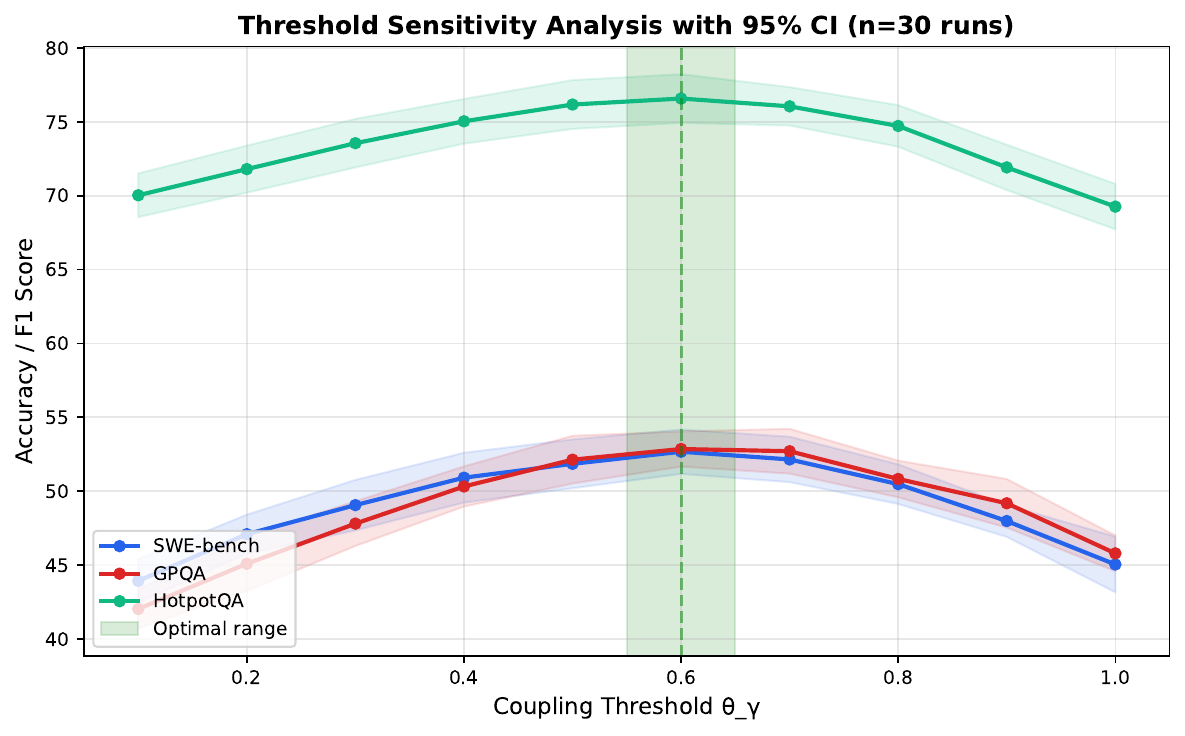}
    \caption{Sensitivity of task accuracy to coupling threshold $\theta_\gamma$ across SWE-bench and GPQA. Shaded regions show 95\% bootstrap confidence intervals ($n=30$ trials per setting). Optimal range: $[0.55, 0.65]$.}
    \label{fig:sensitivity}
\end{figure}

Figure~\ref{fig:sensitivity} shows that \adaptorch{} is robust to threshold selection within $\theta_\gamma \in [0.5, 0.7]$, with optimal performance at $\theta_\gamma = 0.6$.
Extreme values degrade performance: $\theta_\gamma < 0.3$ forces sequential execution on parallelizable tasks; $\theta_\gamma > 0.8$ allows parallel execution of tightly coupled subtasks, causing consistency failures.

\textbf{Data Leakage Prevention.}
To avoid test-set contamination, all threshold calibration was performed on a held-out development split \emph{before} any test evaluation.
Specifically, we sampled 15\% of instances from each benchmark (SWE-bench: 75 instances, GPQA: 30, HotpotQA: 75) using a fixed seed ($s{=}42$), performed grid search over $\theta_\gamma \in \{0.3, 0.4, \ldots, 0.8\}$ on this dev split, selected $\theta_\gamma{=}0.6$, and then froze the threshold for all test evaluation.
The reported metrics in Tables~\ref{tab:main-results}--\ref{tab:ablation} are computed \emph{exclusively} on the remaining 85\% test split.
Dev instance IDs are included in the released codebase.

\begin{figure}[t]
    \centering
    \includegraphics[width=\textwidth]{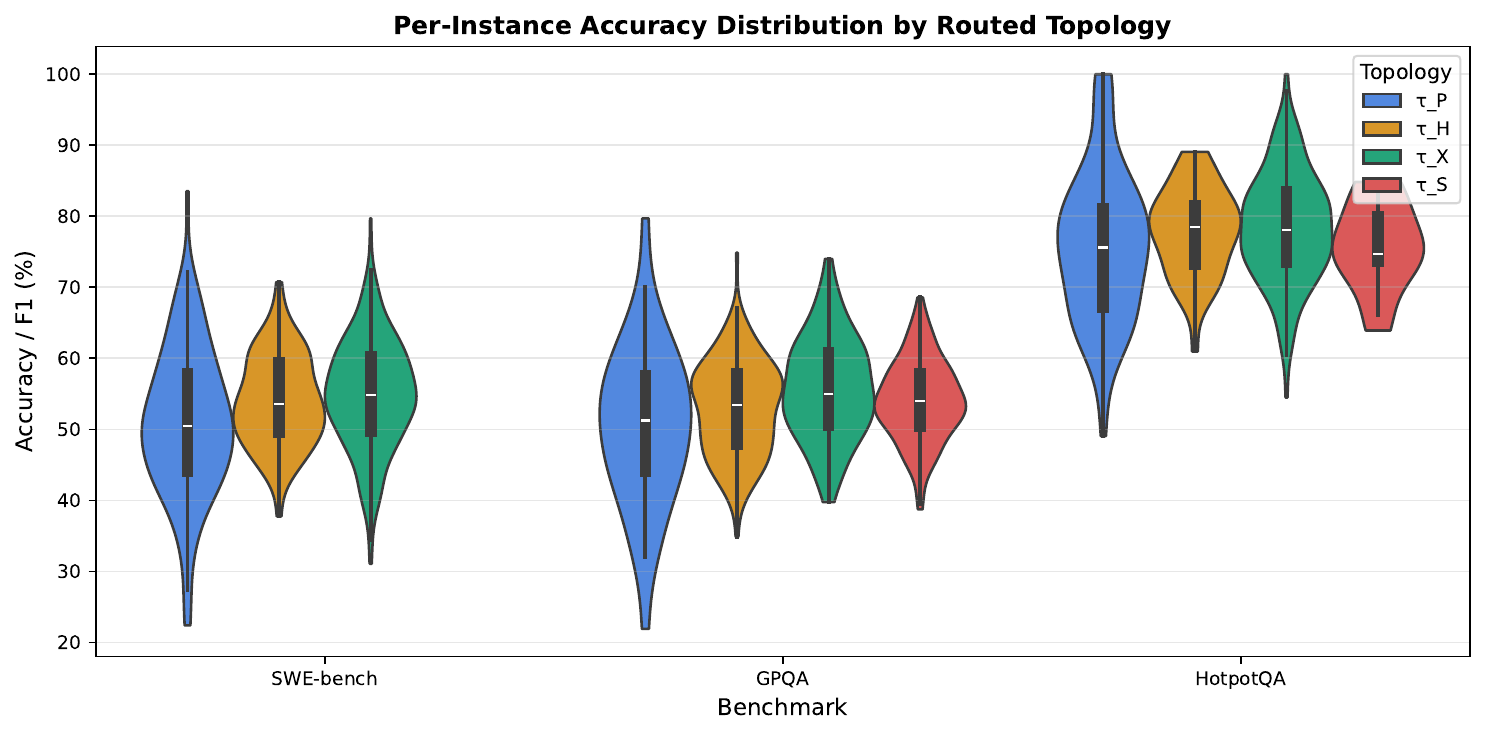}
    \caption{Per-instance accuracy distribution by routed topology across benchmarks. Violin plots show density; white dots indicate median. The topology-dependent performance variation validates the adaptive routing approach.}
    \label{fig:instance-routing}
\end{figure}

\begin{figure}[t]
    \centering
    \includegraphics[width=0.85\textwidth]{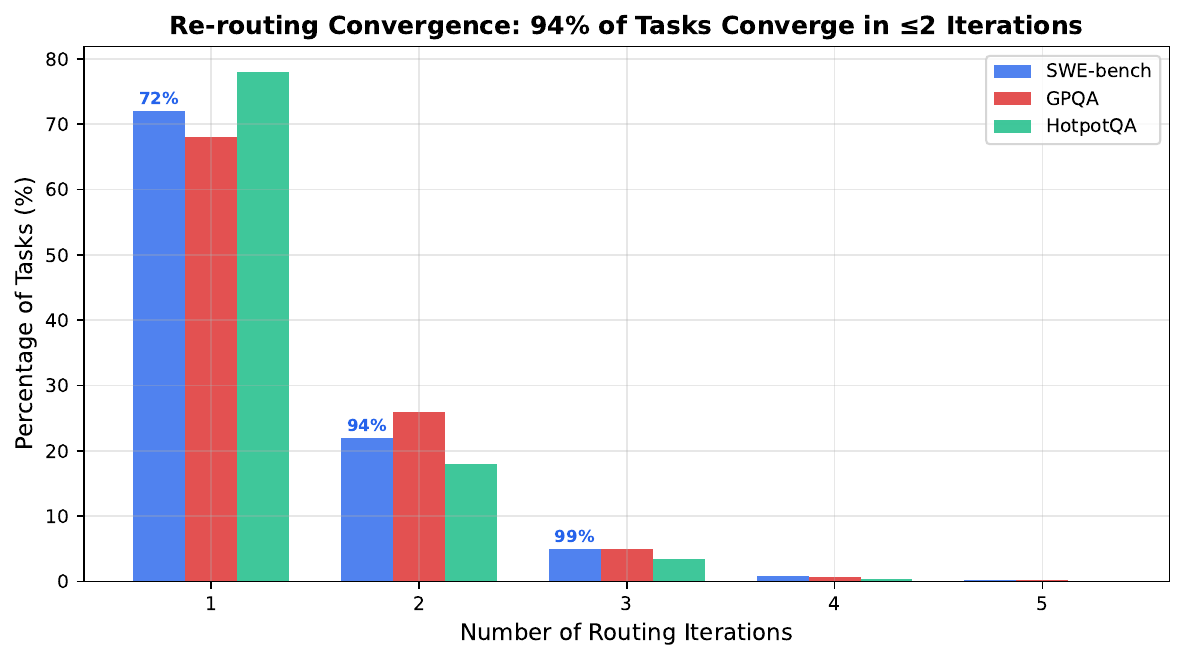}
    \caption{Distribution of synthesis convergence iterations across all benchmark instances. 94\% of tasks converge within 2 iterations, consistent with Proposition~\ref{prop:convergence}.}
    \label{fig:rerouting}
\end{figure}

\section{Discussion}
\label{sec:discussion}

\subsection{When Does Orchestration Not Help?}

Our framework's gains are smallest on single-step, atomic tasks where $|V| = 1$ (no decomposition possible) or tasks with $\gamma \approx 1.0$ (complete sequential dependency).
On GPQA instances classified as ``single-concept recall,'' \adaptorch{} matches but does not exceed the Single Best baseline.
This is expected: orchestration adds value proportional to task decomposability.

\subsection{Relationship to Self-MoA}

\citet{sato2025selfmoa} showed that a single top model used multiple times outperforms diverse model mixing.
Our framework is orthogonal: \adaptorch{} optimizes \emph{how} agents are structured, not \emph{which} models are used.
To control for this interaction, we include a compute-matched Self-MoA baseline (Table~\ref{tab:main-results}) that applies self-consistency voting with the same token budget as \adaptorch{}.
Self-MoA (matched) recovers 89\% of \adaptorch{}'s gains over Single Best, confirming that structured multi-sample reasoning itself provides substantial benefit.
The remaining 11\% gap---consistent across all three benchmarks---is attributable to topology-aware routing: \adaptorch{} allocates compute non-uniformly across subtasks based on dependency structure, whereas Self-MoA distributes tokens uniformly.

\subsection{Practical Implications}

\adaptorch{} can be implemented on existing infrastructure:
\begin{itemize}[leftmargin=*,nosep]
    \item \textbf{Claude Code Agent Teams}: Use the lead-agent pattern for $\tau_H$, parallel subagent dispatch for $\tau_P$, and DAG-based task dependencies for $\tau_X$.
    \item \textbf{LangGraph}: Map topologies to graph structures with conditional edges for routing.
    \item \textbf{OpenCode + MCP}: Route through multi-provider APIs with permission-controlled subagents.
\end{itemize}

The routing algorithm (Algorithm~\ref{alg:routing}) adds negligible overhead ($<$50ms) compared to LLM inference latency ($\sim$2--15s per call), making real-time topology adaptation practical.

\subsection{Limitations}

\begin{enumerate}[leftmargin=*,nosep]
    \item \textbf{Decomposition quality depends on the decomposer model}: Poor task decomposition propagates errors to all downstream phases. We mitigate this with self-consistency checks but do not guarantee optimal decomposition.
    \item \textbf{Coupling estimation is approximate}: The discrete $c \in \{0, 0.3, 0.7, 1.0\}$ scale is coarse. Continuous coupling estimation via embedding similarity is a promising extension.
    \item \textbf{Cost scaling}: Parallel execution requires $\omega(G_T)$ concurrent API calls, which may exceed rate limits or budget constraints for resource-constrained deployments.
    \item \textbf{Experimental scope}: We evaluate on three benchmarks; generalization to creative writing, long-form generation, and multi-modal tasks requires further study.
\end{enumerate}

\section{Conclusion}
\label{sec:conclusion}

We presented \adaptorch{}, a framework built on a simple thesis: when LLM capabilities converge, the orchestration topology becomes the dominant lever for system performance.
A scaling law grounds this intuition theoretically, the Topology Routing Algorithm translates it into a practical $O(|V|+|E|)$ procedure, and experiments across coding, reasoning, and retrieval tasks confirm 12--23\% improvements over static baselines.

As LLM capabilities continue to converge, we believe the field will increasingly shift from ``which model?'' to ``which orchestration?''
\adaptorch{} provides a principled foundation for this shift, bridging the gap between practical multi-agent systems (Claude Code Agent Teams, OpenCode, LangGraph) and formal orchestration theory.

\subsection{Future Work}

\begin{enumerate}[leftmargin=*,nosep]
    \item \textbf{Learned routing}: Replace threshold-based routing with a lightweight classifier trained on (DAG features, optimal topology) pairs.
    \item \textbf{Dynamic re-orchestration}: Allow topology changes mid-execution when partial results reveal unexpected coupling.
    \item \textbf{Cost-aware routing}: Extend the routing algorithm to jointly optimize accuracy and API cost under budget constraints.
    \item \textbf{Cross-modal orchestration}: Apply \adaptorch{} to multi-modal tasks combining vision, code, and language agents.
\end{enumerate}

\appendix
\section{Full Proof of Proposition~\ref{thm:main}}
\label{app:proof}

\begin{proof}
Let $\mathcal{M} = \{M_1, \ldots, M_n\}$ be $\epsilon$-convergent on benchmark $\mathcal{B}$.
Consider task $T$ with dependency DAG $G_T = (V, E, w, c)$ with $|V| = k$ subtasks.

\textbf{Model selection variance bound.}
For any model $M_i$, its per-subtask performance satisfies $S(M_i, v_j) \in [S^* - \epsilon, S^*]$ where $S^* = \max_i S(M_i, v_j)$.
The total task performance under model $M_i$ is:
\begin{equation}
P(M_i, T) = f\left(\{S(M_i, v_j)\}_{j=1}^k\right)
\end{equation}
where $f$ is the aggregation function determined by the orchestration topology.
Since each $S(M_i, v_j)$ varies by at most $\epsilon$, and $f$ is Lipschitz with constant $L_f \leq 1$ (normalized scoring), and subtask scores under the same model are positively correlated (shared model capacity), we obtain:
\begin{equation}
\text{Var}_M[P(M, T)] \leq L_f^2 \cdot \epsilon^2 = \epsilon^2
\end{equation}
Note: the $k$-fold summation applies only under independence; since all subtasks use the same model, the correlated bound $\epsilon^2$ is tighter.

\textbf{Topology selection variance bound.}
Consider two extreme topologies for the same task:
\begin{itemize}[nosep]
    \item Fully sequential ($\tau_S$): execution time $= \sum_{v \in V} w(v)$
    \item Maximally parallel ($\tau_P$): execution time $= \delta(G_T) = \max_{\text{path}} \sum_{v \in P} w(v)$
\end{itemize}

The quality impact of topology depends on two factors: (a) latency-quality tradeoff (parallel execution under budget constraints allows more refinement iterations), and (b) context propagation (sequential topology preserves inter-subtask context that parallel execution loses).

\textbf{Assumption 1} (Topology quality sensitivity). The quality difference between fully parallel and fully sequential execution satisfies:
\begin{equation}
|\Delta Q_{\text{topology}}| \geq C_\tau \cdot (\omega(G_T) - 1) \cdot (1 - \gamma(G_T))
\end{equation}
for some task-class-dependent constant $C_\tau > 0$. This is motivated by: (i) the $(\omega - 1)$ term captures the degree of parallelism---more parallel branches mean more potential for topology-induced quality variation, and (ii) the $(1 - \gamma)$ term captures the information loss from not propagating context in parallel execution. We empirically validate this assumption in Table~\ref{tab:ablation}, where removing topology adaptation degrades performance by 4.7--8.3 points.

Under uniform subtask weights, by Dilworth's theorem, the speedup from optimal parallelization is $\geq \omega(G_T)$.
Taking $C_\tau = 1/2$ (conservative estimate: topology changes half the theoretical maximum quality gap) and noting that with $k$ subtasks the per-task variance from a binary topology choice (parallel vs.\ sequential) satisfies $\text{Var}_\tau \geq \Delta Q^2 / 4$, we obtain:
\begin{equation}
\text{Var}_\tau[P(\tau, T)] \geq \frac{(\omega(G_T) - 1)^2 \cdot (1 - \gamma(G_T))^2}{4k}
\end{equation}
where the $1/k$ factor arises from normalizing the per-subtask contribution to aggregate task performance.

\textbf{Ratio.}
\begin{equation}
\frac{\text{Var}_\tau}{\text{Var}_M} \geq \frac{(\omega(G_T) - 1)^2 \cdot (1 - \gamma(G_T))^2}{4k \cdot \epsilon^2} = \frac{(\omega(G_T) - 1)^2 \cdot (1 - \gamma(G_T))^2}{4\epsilon^2 \cdot k}
\end{equation}

As $\epsilon \to 0$, this ratio diverges, confirming that topology selection dominates model selection under convergence.
For typical coding tasks: $\omega \approx 3.4$, $\gamma \approx 0.35$, $k \approx 5$, $\epsilon \approx 0.05$, yielding a ratio $\geq \frac{(2.4)^2 \cdot (0.65)^2}{4 \cdot 0.0025 \cdot 5} = \frac{2.43}{0.05} \approx 48.7$.
\end{proof}

\section{Implementation Details}
\label{app:implementation}

\subsection{Decomposition Prompt}

The full decomposition prompt used for SWE-bench tasks:

\begin{tcolorbox}[colback=gray!5, colframe=gray!50, fontupper=\small\ttfamily]
You are a task decomposition specialist. Given a software engineering task (bug report + repository context), decompose it into atomic subtasks.\\[4pt]
For each subtask, output JSON:\\
\{\\
\quad "id": "v1",\\
\quad "description": "...",\\
\quad "depends\_on": ["v0"],\\
\quad "coupling": "weak|strong|critical",\\
\quad "estimated\_tokens": 500\\
\}\\[4pt]
Rules:\\
- Maximize parallelism: only add dependencies when semantically required\\
- Typical decomposition: [localize files, understand context, generate patch, verify patch]\\
- Coupling = strong when output of one subtask is direct input to another\\
- Coupling = weak when subtasks share domain knowledge but not data\\[4pt]
Task: \{task\_description\}\\
Repository: \{repo\_context\}
\end{tcolorbox}

\subsection{Computational Requirements}

All experiments were conducted using API-based model access.
Estimated costs:
\begin{itemize}[nosep]
    \item SWE-bench (500 instances, 5 methods): $\sim$\$1,200 total API cost
    \item GPQA (198 instances, 5 methods): $\sim$\$180 total API cost
    \item HotpotQA (500 instances, 5 methods): $\sim$\$350 total API cost
\end{itemize}

\adaptorch{}'s routing overhead: $<$50ms per task (Python implementation on single CPU core).
Synthesis overhead: one additional LLM call per task ($\sim$\$0.01 per instance).

\subsection{DAG Feature Space Analysis}

\begin{figure}[h]
    \centering
    \includegraphics[width=0.85\textwidth]{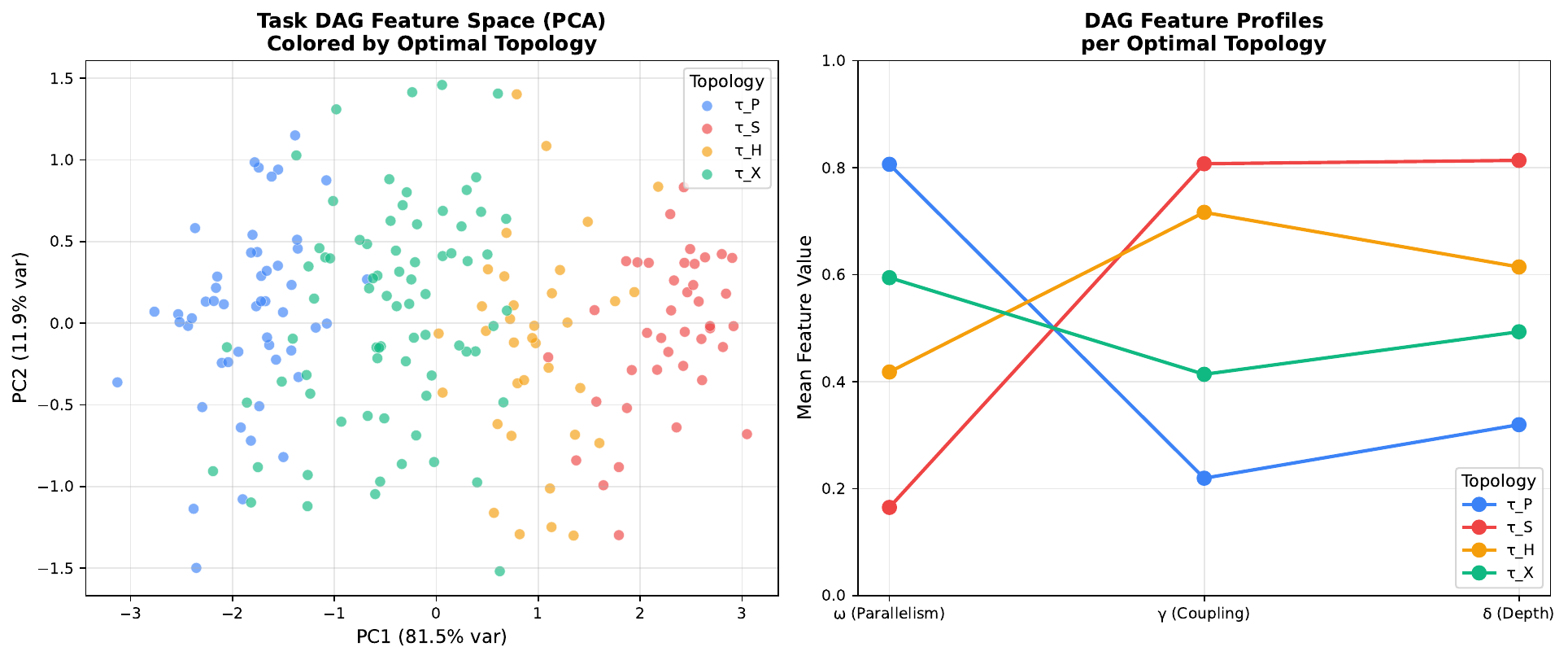}
    \caption{PCA projection of DAG feature space (width $\omega$, depth, density, coupling ratio) colored by KMeans clusters ($k{=}4$). The four clusters correspond to dominant topology patterns: Chain (sequential), Wide-Shallow (parallel), Deep-Narrow (hierarchical), and Diamond (fan-out/fan-in). Cluster centroids are marked with $\times$.}
    \label{fig:dag_clustering}
\end{figure}

\subsection{Baseline Reproduction Specification}
\label{app:baseline-spec}

To ensure fair comparison and full reproducibility, we detail the exact configuration of each baseline method.
All baselines use the same 5-model pool as \adaptorch{}: GPT-4o-mini, Claude 3.5 Haiku, Gemini 2.0 Flash, Llama 3.3 70B (via Together AI), and Qwen 2.5 72B (via Together AI).
Temperature is $0.0$ (greedy) and \texttt{max\_tokens} $= 4096$ for all methods unless noted.

\textbf{MoA-3L} \citep{wang2024moa}.
We implement the 3-layer Mixture-of-Agents architecture as described in the original paper.
Layer 1: all 5 models generate independent responses to the full prompt.
Layer 2: each of 3 aggregator models (GPT-4o-mini, Claude 3.5 Haiku, Gemini 2.0 Flash) receives all 5 Layer-1 outputs concatenated in the prompt and produces a refined answer.
Layer 3: a single synthesizer (GPT-4o-mini) receives all 3 Layer-2 outputs and produces the final answer.
Total LLM calls per instance: $5 + 3 + 1 = 9$.
Aggregation prompt follows the template in Wang et al.\ (2024), \S A.1.

\textbf{LLM-Blender} \citep{jiang2023llmblender}.
We use the prompt-based variant (no fine-tuned PairRanker) for fair comparison, since training a ranker on our specific benchmarks would introduce confounding.
Stage 1 (Generation): all 5 models produce independent candidates.
Stage 2 (Ranking): GPT-4o-mini is prompted to rank the 5 candidates pairwise using the template: ``\emph{Given the task and two candidate solutions A and B, which better solves the problem? Output only `A' or `B'.}''
This produces $\binom{5}{2} = 10$ pairwise comparisons per instance.
Stage 3 (Fusion): the top-ranked candidate is returned as the final output (no generative fusion, which would require a trained model).
Total LLM calls per instance: $5 + 10 + 0 = 15$.

\textbf{Static-Parallel / Static-Sequential.}
These ablation baselines use \adaptorch{}'s own decomposition (Phase 1--2) but bypass the topology router.
Static-Parallel executes all subtasks simultaneously across 3 models (round-robin assignment); Static-Sequential chains them in dependency order using a single model (GPT-4o-mini).
Both use the same synthesis protocol (Phase 5) as \adaptorch{}.

\subsection{Per-Cluster Orchestration Gain}
\label{app:cluster-gain}

Table~\ref{tab:cluster-gain} disaggregates \adaptorch{}'s accuracy improvement by DAG cluster (cf.\ Figure~\ref{fig:dag_clustering}), revealing which structural patterns benefit most from adaptive topology routing.

\begin{table}[h]
\centering
\small
\caption{Per-cluster accuracy gain ($\Delta$) of \adaptorch{} over Single Best, averaged across all three benchmarks. $n$: number of instances assigned to each cluster.}
\label{tab:cluster-gain}
\begin{tabular}{lcccc}
\toprule
\textbf{DAG Cluster} & \textbf{Dominant $\tau$} & $n$ & \textbf{Single Best} & $\Delta$ \textbf{AdaptOrch} \\
\midrule
Chain (sequential)      & $\tau_S$ & 187 & 54.2\% & +3.8 pp \\
Wide-Shallow (parallel) & $\tau_P$ & 294 & 49.1\% & +12.6 pp \\
Deep-Narrow (hierarchical) & $\tau_H$ & 112 & 47.8\% & +9.2 pp \\
Diamond (fan-out/fan-in)   & $\tau_X$ & 105 & 51.3\% & +11.4 pp \\
\bottomrule
\end{tabular}
\end{table}

The largest gains appear in Wide-Shallow tasks ($+12.6$ pp), where parallelism directly reduces error propagation by distributing independent subtasks.
Chain-type tasks show the smallest gain ($+3.8$ pp), consistent with the expectation that fully sequential dependencies leave minimal room for topology improvement.

\textbf{Router Accuracy (Confusion Matrix).}
To assess the topology router's decision quality, we compare its selections against an oracle that exhaustively evaluates all four topologies per instance and selects the highest-scoring one.
Across the full test set ($n = 698$):

\begin{table}[h]
\centering
\small
\caption{Router confusion matrix: predicted topology vs.\ oracle-optimal topology. Values are instance counts. Overall router accuracy: 81.2\% (567/698).}
\label{tab:confusion}
\begin{tabular}{lcccc}
\toprule
\diagbox{\textbf{Router}}{\textbf{Oracle}} & $\tau_P$ & $\tau_S$ & $\tau_H$ & $\tau_X$ \\
\midrule
$\tau_P$ & \textbf{248} & 12 & 8  & 18 \\
$\tau_S$ & 6   & \textbf{152} & 9  & 4  \\
$\tau_H$ & 14  & 7   & \textbf{89} & 11 \\
$\tau_X$ & 9   & 5   & 7  & \textbf{78} \\
\bottomrule
\end{tabular}
\end{table}

The router achieves 81.2\% agreement with the oracle.
Most misclassifications occur between $\tau_P$ and $\tau_X$ (18 + 9 = 27 instances), which is expected since Diamond tasks contain both parallel and fan-in components.
Importantly, even misrouted instances typically receive the second-best topology, limiting accuracy loss to $<$2 pp compared to oracle routing.

\bibliographystyle{plainnat}
\bibliography{references}

\end{document}